\documentclass[11pt]{article}

\usepackage[margin=1in]{geometry}
\usepackage{amsmath,amssymb,amsthm}
\usepackage{mathtools}
\usepackage{bm}
\usepackage{graphicx}
\usepackage{booktabs}
\usepackage{algorithm}
\usepackage{algpseudocode}
\usepackage[colorlinks=true,citecolor=blue,linkcolor=blue,urlcolor=blue]{hyperref}
\usepackage{natbib}
\usepackage{enumitem}
\usepackage{xcolor}

\newtheorem{assumption}{Assumption}
\newtheorem{proposition}{Proposition}
\newtheorem{lemma}{Lemma}

\newtheorem{definition}{Definition}
\newtheorem{remark}{Remark}
\newtheorem*{proof*}{Proof}

\newcommand{\R}{\mathbb{R}}

\newcommand{\argmin}{\operatorname*{arg\,min}}
\DeclareMathOperator{\rank}{rank}
\DeclareMathOperator{\Null}{null}

\title{\textbf{Spectral Truncation in Synthetic Control: \\ Weight Underdetermination and Basis-Estimation Noise}}

\author{
  Mojtaba Eslami \\[4pt]
  \small University of Calgary,  
  \small \texttt{mojtaba.eslami@alumni.ucalgary.ca}
}

\date{}

\date{\today}

\begin{document}
\maketitle

\begin{abstract}
Synthetic control (SC) matches a treated unit's raw pre-treatment trajectory against a weighted combination
of donor units. We study \emph{Spectral SC}, which instead matches in the coordinates of the donor panel's
leading temporal singular vectors, and a \emph{hybrid} estimator that assigns independently tunable weight to
retained and discarded directions, nesting raw-path SC and truncated Spectral SC as endpoints. We prove this
family reduces exactly to raw-path SC at full rank, prove that exact spectral balance with $K$ retained
dimensions and $N_0$ donors is underdetermined whenever $N_0>K+1$ (an affine solution set of dimension
$N_0-K-1$), and give a finite-sample bound relating spectral balancing quality to treatment-effect bias via a
best-linear-predictor decomposition of loadings in spectral scores. We evaluate the family -- with
regularization and mixing weight, but not rank, selected by donor-only placebo validation -- across eleven
data-generating regimes chosen to favor truncation ($400$ replications each, with Monte Carlo standard errors
and paired replication-level comparisons). Truncated Spectral SC has significantly higher RMSE than tuned
raw-path SC in every regime (paired differences $4$--$11$ standard errors from zero); the hybrid estimator
selects raw-path matching outright in a majority of replications throughout and is statistically
indistinguishable from tuned SC in most regimes. A robustness check shows this gap is highly
preprocessing-dependent: under raw input the gap is large, but removing unit and time fixed effects before
the spectral decomposition -- consistent with the assumption underlying our bound -- nearly eliminates it,
and the placebo procedure itself then switches to preferring truncation. We report this as a diagnostic
result, identifying specific mechanisms (basis-estimation noise, balancing underdetermination, and
fixed-effects contamination of the spectral basis) governing when spectral matching can help, rather than
establishing that the version tested here should replace raw-path SC.
\end{abstract}


\section{Introduction}
\label{sec:intro}

Synthetic control (SC) \citep{abadie2003,abadie2010} estimates a treated unit's counterfactual path as a
weighted average of untreated donor units, with weights chosen so the weighted donor pre-treatment trajectory
matches the treated unit's own pre-treatment trajectory. This matching is performed on the raw,
$T_0$-dimensional pre-treatment path.

This paper studies an alternative: matching in the coordinates of the panel's leading temporal singular
vectors rather than in raw time coordinates. The motivating hypothesis is that, if the untreated outcome
process follows a low-rank interactive fixed-effects model, explicitly matching in the directions estimated
to carry the panel's systematic variation might target the underlying factor loadings more directly than
matching the full raw trajectory. This is a hypothesis, not a settled conclusion: raw-path matching may
already estimate factor-relevant weights efficiently, since repeated time observations carry information even
when the underlying process is low rank, and projecting onto an estimated subspace first introduces its own
estimation error while discarding whatever implicit regularization the full trajectory provided. We test the
hypothesis with a simulation study designed specifically to find regimes in which it holds, and we
characterize, formally, two specific mechanisms -- weight underdetermination and basis-estimation noise from
unremoved fixed effects -- that can make it fail.

\paragraph{Scope and claims.} Low-rank and spectral structure in panel causal inference is already well
established. Matrix completion \citep{athey2021}, generalized synthetic control \citep{xu2017}, robust
synthetic control \citep{amjad2018}, and augmented synthetic control \citep{benmichael2021} all use low-rank
or factor-model structure to correct or regularize SC-type estimators, and synthetic difference-in-differences
(SDID) \citep{arkhangelsky2021} is analyzed under a latent factor model. Synthetic Principal Component Design
\citep{lu2022} uses a spectral optimization for the experimental-design problem of which units to treat; a
functional extension of generalized synthetic control uses functional principal component scores for sparse,
irregular panels \citep{shao2026}; and Harmonic Synthetic Control \citep{liu2026} introduces a continuously
tunable spectral allocation for a related but distinct purpose (Section~\ref{sec:related}). We do not claim
that using spectral or low-rank structure in panel causal inference is new. What we study, narrowly, is a
simplex-constrained score-balancing estimator and a hybrid generalization, with two new formal results (weight
underdetermination and a score-to-loading bound) and an honest empirical account of when the resulting
truncation helps -- including a preprocessing-dependent reversal that we think is the paper's most useful
empirical finding.

\paragraph{Roadmap.} Section~\ref{sec:setup} sets up notation. Section~\ref{sec:spectral} defines the
spectral and hybrid estimators and proves the full-rank equivalence and weight-underdetermination results.
Section~\ref{sec:identification} gives a formal representability condition and the bias bound.
Section~\ref{sec:preprocessing} addresses how fixed effects are removed before the spectral decomposition and
reports a robustness check that reverses the paper's main conclusion under one preprocessing choice.
Section~\ref{sec:algorithm} describes hyperparameter selection, including which parameters are tuned and
which are fixed by design, and a placebo-sample-size sensitivity check.
Section~\ref{sec:simulation} reports the eleven-regime study with paired statistics.
Section~\ref{sec:discussion} discusses related work and open questions. Appendix~\ref{sec:appendix} gives
reproducibility details.

\section{Setup}
\label{sec:setup}

Let $i=1,\dots,N$ index units and $t=1,\dots,T$ index periods, with $T_0$ pre-treatment and $T_1=T-T_0$
post-treatment periods. Unit $1$ is treated starting at $T_0+1$; units $i=2,\dots,N$ (the donor pool
$\mathcal{C}$, $|\mathcal{C}|=N_0$) are never treated. $Y_{it}=Y_{it}(0)+\tau_{it}W_{it}$, with
$W_{it}\in\{0,1\}$ the treatment indicator, and
\begin{equation}
  Y_{it}(0) = \alpha_i + \delta_t + \ell_i^\top f_t + \varepsilon_{it},
  \label{eq:factor-model}
\end{equation}
with $\alpha_i,\delta_t$ unit and time fixed effects, $\ell_i\in\R^R$ factor loadings, $f_t\in\R^R$ common
factors, $R\ll\min(N,T)$, $\varepsilon_{it}$ idiosyncratic noise -- the standard generating process used to
motivate SC-type estimators \citep{abadie2010,xu2017,athey2021,arkhangelsky2021}.

Stacking donor pre-treatment outcomes gives $Y_{0,\text{pre}}\in\R^{N_0\times T_0}$, approximately low rank
under \eqref{eq:factor-model}. Let $\omega\in\Delta^{N_0}=\{\omega_i\ge0,\sum_i\omega_i=1\}$. Raw-path SC
solves $\hat\omega=\argmin_{\omega\in\Delta^{N_0}}\|y_{1,\text{pre}}-\sum_i\omega_i y_{i,\text{pre}}\|_2^2$,
with $\hat Y_{1t}(0) = \sum_i\hat\omega_i Y_{it}$ for $t>T_0$.

\section{Spectral and Hybrid Balancing}
\label{sec:spectral}

\subsection{Spectral representation and the baseline estimator}

Let $Y_{0,\text{pre}}=U\Sigma V^\top$ be the SVD of the donor pre-treatment matrix, estimated from donor
pre-treatment cells only. Fix $K$ and keep the leading components, $V_K\in\R^{T_0\times K}$. Project every
pre-treatment path onto $V_K$: $s_1=V_K^\top y_{1,\text{pre}}$, $s_i=V_K^\top y_{i,\text{pre}}$ for
$i\in\mathcal{C}$, and solve
\begin{equation}
  \hat\omega = \argmin_{\omega\in\Delta^{N_0}} \Big\| s_1 - \sum_{i=1}^{N_0}\omega_i s_i \Big\|_2^2 +
  \lambda_\omega \|\omega\|_2^2, \qquad
  \hat Y_{1t}(0) = \sum_{i=1}^{N_0} \hat\omega_i Y_{it} \ \ (t>T_0).
  \label{eq:spectral-sc}
\end{equation}

\begin{proposition}[Equivalence to raw-path SC at full rank]
\label{prop:equiv}
Suppose $T_0\le N_0$ and $K=T_0$, so $V_K=V\in\R^{T_0\times T_0}$ is square and orthogonal. Then for every
$\omega\in\Delta^{N_0}$, $\|V_K^\top(y_{1,\text{pre}}-Y_{0,\text{pre}}^\top\omega)\|_2^2 =
\|y_{1,\text{pre}}-Y_{0,\text{pre}}^\top\omega\|_2^2$, so \eqref{eq:spectral-sc} with $K=T_0$ has exactly the
same objective, feasible set, and minimizer as raw-path SC with the same ridge penalty.
\end{proposition}
\begin{proof*}
Let $x=y_{1,\text{pre}}-Y_{0,\text{pre}}^\top\omega$. Since $V$ is orthogonal, $\|V^\top x\|_2^2 = x^\top
VV^\top x = x^\top x = \|x\|_2^2$, for every $\omega$; the objectives coincide pointwise on $\Delta^{N_0}$.
\qed
\end{proof*}

This shows the projection preserves the raw Euclidean pre-treatment balancing discrepancy at $K=T_0$; it does
not, by itself, establish equivalence to any estimator built from more than this discrepancy (see the remark
on SDID in Section~\ref{sec:hybrid}). Any difference between \eqref{eq:spectral-sc} and raw-path SC for
$K<T_0$ comes specifically from truncating modes.

\subsection{Weight underdetermination}
\label{sec:underdetermination}

The title's second mechanism is a precise linear-algebra fact about \eqref{eq:spectral-sc}, not merely a
qualitative observation. Consider the equality-constrained version of the spectral balancing problem
obtained by dropping the ridge penalty and the sign constraints on $\omega$ (i.e.\ exact balance on the
affine hull of the simplex): with donor score matrix $S_K=[s_2\ \cdots\ s_{N_0}]\in\R^{K\times N_0}$, exact
spectral balance requires
\begin{equation}
  S_K\omega = s_1, \qquad \mathbf{1}^\top\omega = 1,
  \label{eq:exact-balance}
\end{equation}
a system of $K+1$ linear equations in $N_0$ unknowns.

\begin{proposition}[Underdetermination of low-rank balancing]
\label{prop:underdetermination}
Let $A=\begin{bmatrix}S_K\\ \mathbf{1}^\top\end{bmatrix}\in\R^{(K+1)\times N_0}$. If $N_0>K+1$, system
\eqref{eq:exact-balance} is consistent (i.e.\ has at least one solution), and $A$ has full row rank $K+1$,
then the solution set $\{\omega\in\R^{N_0}: A\omega = [s_1;1]\}$ is a nonempty affine subspace of dimension
$N_0-K-1$.
\end{proposition}
\begin{proof*}
By the rank-nullity theorem, $\dim\Null(A) = N_0-\rank(A) = N_0-(K+1)$. If $\omega_0$ is one solution, the
full solution set is $\omega_0+\Null(A)$, an affine translate of $\Null(A)$, hence of the same dimension. \qed
\end{proof*}

Whenever $N_0>K+1$ -- generically true, since $K$ is chosen small by design while $N_0$ is the size of the
available donor pool -- exact spectral balance does not pin down a unique donor weight vector: an entire
$(N_0-K-1)$-dimensional family of weight vectors achieves identical pre-treatment spectral balance while
generally implying \emph{different} post-treatment counterfactuals, since nothing in \eqref{eq:exact-balance}
constrains behavior outside $\mathrm{range}(V_K)$. The ridge penalty $\lambda_\omega\|\omega\|_2^2$ (together
with the simplex sign constraints, which cut down but do not generally eliminate this indeterminacy) is what
selects a particular point from this set; as $N_0-K-1$ grows, that selection carries correspondingly more of
the estimator's effective information content, since the data no longer pin it down. This is the formal
counterpart of the empirical finding in Section~\ref{sec:simulation}: with $N_0=30$ and $K=2$, the affine
solution set of \eqref{eq:exact-balance} has dimension up to $27$, and different regularization strengths can
select donor weight vectors with materially different post-treatment behavior even though all of them
balance the retained two-dimensional score equally well. Raw-path SC, with $K=T_0$, faces the analogous
system with $T_0+1$ equations rather than $K+1$; whenever $T_0$ is not much smaller than $N_0$ this system is
far less underdetermined, and in the $T_0\gg N_0$ regime of Section~\ref{sec:simulation} it can be
overdetermined, which is one reason raw-path matching is implicitly well regularized in that regime without
any explicit penalty.

\subsection{A hybrid estimator}
\label{sec:hybrid}

A less abrupt alternative to truncation assigns discarded directions a strictly positive but shrunk weight
$\eta\in[0,1]$, using $\Pi_K=V_KV_K^\top$:
\begin{equation}
  \hat\omega_\eta = \argmin_{\omega\in\Delta^{N_0}} (y_{1,\text{pre}}-Y_{0,\text{pre}}^\top\omega)^\top
  M_{K,\eta} (y_{1,\text{pre}}-Y_{0,\text{pre}}^\top\omega) + \lambda_\omega\|\omega\|_2^2, \qquad
  M_{K,\eta} = \Pi_K + \eta(I_{T_0}-\Pi_K).
  \label{eq:hybrid}
\end{equation}
By Proposition~\ref{prop:equiv}'s argument, $\eta=1$ recovers raw-path SC exactly for any $K$; $\eta=0$
recovers truncated Spectral SC \eqref{eq:spectral-sc} exactly. Values $\eta\in(0,1)$ interpolate, and by
Proposition~\ref{prop:underdetermination}, moving $\eta$ above $0$ directly reduces the dimension of
directions left unconstrained by the balancing loss -- $\eta>0$ constrains all $T_0$ directions, just with
unequal weight, rather than leaving $T_0-K$ of them entirely unconstrained. This turns truncation into a
continuous parameter tunable by the same placebo procedure used for the other estimators
(Section~\ref{sec:algorithm}).

\begin{remark}[On recovering SDID]
At full rank with identity mode metrics, \eqref{eq:spectral-sc}--\eqref{eq:hybrid} preserve the raw Euclidean
balancing discrepancy (Proposition~\ref{prop:equiv}). Recovering the complete SDID estimator
\citep{arkhangelsky2021} additionally requires its precise unit-weight, time-weight, intercept,
regularization, and final doubly-weighted regression specification; matching one Euclidean discrepancy term
is necessary but not sufficient for that equivalence, and we do not pursue a time-weighted analogue here.
\end{remark}

\section{Identification and Bias}
\label{sec:identification}

\subsection{A formal representability condition}

We use the following as a motivating condition rather than a claim we verify: it states what would need to be
true for spectral balancing to be well targeted, in terms of explicit, interpretable tolerances.

\begin{assumption}[Approximate rank-$K$ representability]
\label{ass:spectral-pt}
Fix tolerances $\delta_s,\delta_\ell\ge0$. There exists $\omega^\star\in\Delta^{N_0}$ such that
\begin{equation}
  \big\|\Pi_K\big(y_{1,\text{pre}}-Y_{0,\text{pre}}^\top\omega^\star\big)\big\|_2 \le \delta_s,
  \qquad
  \Big\|\ell_1-\sum_{i}\omega^\star_i\ell_i\Big\|_2 \le \delta_\ell,
  \label{eq:formal-assumption}
\end{equation}
where $\Pi_K,\{\ell_i\}$ are the same population objects (not re-estimated) across the pre- and post-treatment
windows.
\end{assumption}

Smaller $(\delta_s,\delta_\ell)$ is a stronger, more favorable condition; $\delta_s=\delta_\ell=0$ recovers
exact spectral balance with exact loading representability. Under Assumption~\ref{ass:spectral-pt}, the
loading-mismatch term in Proposition~\ref{prop:bias} below is bounded by $\delta_\ell\|f_t\|_2$ at the oracle
$\omega^\star$; Proposition~\ref{prop:bound} in Section~\ref{sec:bound} relates the \emph{achieved}
estimator's $\hat\omega$ (which need not equal $\omega^\star$) to the \emph{observed} spectral residual
$\|s_1-\sum_i\hat\omega_is_i\|_2$, which is what an analyst can actually compute.

\begin{proposition}[Full-path balance implies rank-$K$ subspace balance]
\label{prop:nesting}
If $\omega\in\Delta^{N_0}$ satisfies $y_{1,\text{pre}}=\sum_i\omega_iy_{i,\text{pre}}$ exactly, it satisfies
$\Pi_K y_{1,\text{pre}}=\Pi_K\sum_i\omega_iy_{i,\text{pre}}$ for every $K,\Pi_K$.
\end{proposition}
\begin{proof*}
Apply $\Pi_K$ to both sides. \qed
\end{proof*}

The converse fails whenever the residual has a nonzero component outside $\mathrm{range}(V_K)$, so
Assumption~\ref{ass:spectral-pt} at a given $\delta_s$ is different from, not proven weaker than, the
corresponding full-path condition.

\subsection{Exact bias decomposition}

\begin{proposition}[Bias decomposition under the factor model]
\label{prop:bias}
Under \eqref{eq:factor-model}, for any $\omega\in\Delta^{N_0}$ and $t>T_0$, with $\hat
Y_{1t}(0)=\sum_{i\in\mathcal{C}}\omega_iY_{it}$ and $\hat\tau_t := Y_{1t}-\hat Y_{1t}(0)$,
\begin{equation}
  \hat\tau_t-\tau_t = \underbrace{\Big(\alpha_1-\sum_i\omega_i\alpha_i\Big)}_{\text{level mismatch}} +
  \underbrace{\Big(\ell_1-\sum_i\omega_i\ell_i\Big)^\top f_t}_{\text{loading mismatch}} +
  \underbrace{\Big(\varepsilon_{1t}-\sum_i\omega_i\varepsilon_{it}\Big)}_{\text{idiosyncratic noise}}.
  \label{eq:bias-decomp}
\end{equation}
\end{proposition}
\begin{proof*}
Substitute \eqref{eq:factor-model} (with $Y_{1t}$ additionally containing $\tau_t$) into
$\hat\tau_t=Y_{1t}-\sum_i\omega_iY_{it}$ and use $\sum_i\omega_i=1$ to cancel $\delta_t$. \qed
\end{proof*}

\subsection{From spectral score balance to loading mismatch}
\label{sec:bound}

Proposition~\ref{prop:bias} isolates the loading-mismatch term as what spectral balancing targets. We now
bound it in terms of the observable spectral residual, using a best-linear-predictor (BLP) construction so
that the relevant linear map is defined, not merely assumed to exist.

\begin{definition}[BLP of loadings on spectral scores]
\label{def:blp}
Let $A_K\in\R^{R\times K}$ be the least-squares (best linear predictor) coefficient obtained by regressing
$\{\ell_i\}_{i\in\{1\}\cup\mathcal{C}}$ on $\{s_i\}_{i\in\{1\}\cup\mathcal{C}}$: $A_K :=
\argmin_{A}\sum_i\|\ell_i-As_i\|_2^2$, and let $r_i:=\ell_i-A_Ks_i$ be the resulting residuals. $A_K$ is well
defined whenever the $s_i$ span $\R^K$ (generic whenever $N_0\ge K$), with no distributional or
correct-specification assumption required.
\end{definition}

\begin{lemma}[$A_K$ is exact when $K=R$ and $V_K$ spans the true factor subspace]
\label{lem:linear-map}
If unit and time fixed effects have been removed, so $y_{i,\text{pre}}=F_{\text{pre}}\ell_i+
\varepsilon_{i,\text{pre}}$, and $K=R$ with $A:=V_K^\top F_{\text{pre}}$ invertible, then the population BLP
coincides with $A_K=A^{-1}$ and $r_i=-A^{-1}V_K^\top\varepsilon_{i,\text{pre}}$.
\end{lemma}
\begin{proof*}
$s_i=A\ell_i+V_K^\top\varepsilon_{i,\text{pre}}$, so $\ell_i=A^{-1}s_i-A^{-1}V_K^\top\varepsilon_{i,\text{pre}}$
exactly; since this linear relation holds exactly for every $i$, it is in particular the least-squares
solution. \qed
\end{proof*}

\begin{proposition}[Error decomposition and observable upper bound under approximate score representability]
\label{prop:bound}
For $A_K,r_i$ as in Definition~\ref{def:blp} and any $\omega\in\Delta^{N_0}$,
\begin{equation}
  \Big\|\ell_1-\sum_i\omega_i\ell_i\Big\|_2 \;\le\; \|A_K\|_{\mathrm{op}}\,
  \Big\|s_1-\sum_i\omega_is_i\Big\|_2 \;+\; \Big\|r_1-\sum_i\omega_ir_i\Big\|_2,
  \label{eq:bound}
\end{equation}
and, at any $t>T_0$, $\big|(\ell_1-\sum_i\omega_i\ell_i)^\top f_t\big| \le \|f_t\|_2\big[\|A_K\|_{\mathrm{op}}
\|s_1-\sum_i\omega_is_i\|_2 + \|r_1-\sum_i\omega_ir_i\|_2\big]$.
\end{proposition}
\begin{proof*}
$\ell_1-\sum_i\omega_i\ell_i = A_K(s_1-\sum_i\omega_is_i)+(r_1-\sum_i\omega_ir_i)$ by construction and
$\sum_i\omega_i=1$; triangle inequality and submultiplicativity give the first bound, Cauchy--Schwarz the
second. \qed
\end{proof*}

This inequality is close to algebraic once Definition~\ref{def:blp} is in place, and we do not present it as
a deep identification theorem; its value is interpretive. It makes precise (i) that spectral balancing
controls loading mismatch \emph{only through} $\|A_K\|_{\mathrm{op}}$ times the achieved spectral residual --
exactly the quantity \eqref{eq:spectral-sc} minimizes -- and (ii) that the BLP residual term $r_i$, which
absorbs both omitted-factor error ($K<R$) and basis-estimation error, is invisible to the spectral objective
and not controlled by it at all. When $r_i$ is large, \eqref{eq:bound} is uninformative regardless of how well
\eqref{eq:spectral-sc} is solved -- consistent with the simulation evidence in
Sections~\ref{sec:preprocessing}--\ref{sec:simulation}.

\section{Preprocessing and the Choice of Basis}
\label{sec:preprocessing}

Lemma~\ref{lem:linear-map} assumes unit and time fixed effects have already been removed before the SVD. The
practical estimator in Section~\ref{sec:spectral}, and the simulation study as originally run, applies the
SVD to \emph{raw} pre-treatment paths. This matters: if $\alpha_i$ varies substantially across donors, the
leading singular directions of $Y_{0,\text{pre}}$ can be dominated by cross-unit level differences rather
than by the interactive factor structure in \eqref{eq:factor-model}, in which case $V_K$ estimated from raw
data is a poor estimate of the population factor subspace and the BLP residual $r_i$ in
Proposition~\ref{prop:bound} is correspondingly large through no fault of the truncation rule itself.

We evaluate three preprocessing choices, applied identically to raw-path SC, Spectral SC, and the hybrid
estimator so the comparison stays apples-to-apples:
\begin{itemize}[leftmargin=1.6em,itemsep=1pt]
  \item \textbf{Raw:} balancing and basis estimation use $Y_{0,\text{pre}},y_{1,\text{pre}}$ directly (the
        estimator as defined in Section~\ref{sec:spectral} and used in the main study).
  \item \textbf{Unit-demeaned:} each unit's own pre-treatment mean is subtracted before balancing and basis
        estimation; the counterfactual is reconstructed as $\sum_i\hat\omega_iY_{it}$ plus an intercept
        correction $\bar y_{1,\text{pre}}-\sum_i\hat\omega_i\bar y_{i,\text{pre}}$, so the estimator remains a
        donor-weighted average with a level adjustment, not a forecast.
  \item \textbf{Two-way demeaned:} unit means and donor-pool time means are both removed (double centering)
        before balancing and basis estimation, with the same intercept correction as above.
\end{itemize}

\begin{table}[h]
\centering
\small
\begin{tabular}{lrrrrc}
\toprule
Preprocessing & SC RMSE & Spectral RMSE & $\Delta$(Spectral$-$SC) & $\bar\eta$ & $\Pr(\hat\eta{=}0)$ \\
\midrule
Raw              & 0.283 (0.019) & 0.431 (0.025) & 0.148 (0.025) & 0.959 & 0.008 \\
Unit-demeaned    & 0.213 (0.012) & 0.238 (0.016) & 0.025 (0.012) & 0.878 & 0.032 \\
Two-way demeaned & 0.223 (0.013) & 0.224 (0.013) & 0.001 (0.004) & 0.295 & 0.568 \\
\bottomrule
\end{tabular}
\caption{Preprocessing robustness check, baseline regime, $N_0=30$, $K=2$, $250$ replications per row.
Standard errors in parentheses; $\Delta$ is the paired RMSE difference (same simulated panels).}
\label{tab:preprocessing}
\end{table}

Table~\ref{tab:preprocessing} shows a reversal, not merely an attenuation. Under raw input, Spectral SC's
RMSE exceeds tuned SC's by $0.148$ ($\mathrm{SE}\ 0.025$), about six standard errors -- consistent with the
main study in Section~\ref{sec:simulation}, which uses raw input throughout. Under unit-demeaning the gap
falls to $0.025$ ($\mathrm{SE}\ 0.012$), about two standard errors. Under two-way demeaning the gap is $0.001$
($\mathrm{SE}\ 0.004$), statistically indistinguishable from zero, and -- more strikingly -- the
placebo-tuning procedure itself switches its preference: the mean selected $\eta$ falls from $0.96$ (raw) to
$0.30$ (two-way demeaned), with a majority of replications ($57\%$) now selecting $\eta=0$, full truncation,
rather than $\eta=1$. This is exactly the pattern Proposition~\ref{prop:bound} predicts: removing fixed
effects before the SVD shrinks the BLP residual $r_i$ by bringing $V_K$ closer to the true factor subspace,
and the placebo procedure -- which never sees this proposition, only post-treatment squared error on held-out
donors -- detects the resulting improvement in $\|A_K\|_{\mathrm{op}}$-scaled control on its own.

We do not rerun the full eleven-regime study of Section~\ref{sec:simulation} under two-way demeaning; doing
so, and mapping out exactly which regimes benefit and by how much, is a natural next step this note does not
complete. We report Table~\ref{tab:preprocessing} as a robustness check on the baseline regime and flag
explicitly that the negative results in Section~\ref{sec:simulation} are for raw-path input specifically, not
for spectral matching under every reasonable preprocessing choice.

\section{Hyperparameter Selection}
\label{sec:algorithm}

We tune $\lambda_\omega$, and, for the hybrid estimator, $\eta$, by donor-only leave-one-out placebo
validation: for each donor $j$ in a randomly drawn subset, we treat $j$ as if treated, re-estimate the basis
from the remaining donors, balance $j$'s pre-treatment path against them at each candidate
$(\lambda_\omega,\eta)$, and record squared error on $j$'s actual post-treatment path, never using
post-treatment or treated-unit data. We use the grids $\lambda_\omega\in\{10^{-4},10^{-3},10^{-2},10^{-1},1\}$
and $\eta\in\{0,0.15,0.35,0.5,0.65,0.85,1\}$.

\paragraph{$K$ is fixed by design, not tuned.} Throughout this paper, the rank $K$ is a prespecified input to
Algorithm~\ref{alg:spectral-did}, chosen per regime as part of the experimental design (e.g.\ set equal to,
below, or above the data-generating process's true factor rank $R$, to study misspecification directly, as in
the weak-factor regimes of Section~\ref{sec:simulation}). We tune only $(\lambda_\omega,\eta)$ by placebo
validation. We do not fold $K$ into the same procedure because the placebo objective is comparable across
candidate $\lambda_\omega,\eta$ values at fixed $K$ (all evaluated in the same $T_0$-dimensional held-out
squared-error units) but is not obviously comparable across $K$ within a single, unified selection rule
without additional structure; treating $K$ as fixed avoids that complication and keeps the object of study --
the consequence of a given truncation choice -- separate from the object of the search.

\begin{algorithm}[h]
\caption{Spectral / hybrid SC with placebo-tuned $(\lambda_\omega,\eta)$ at fixed $K$}
\label{alg:spectral-did}
\begin{algorithmic}[1]
\Require Donor pre-treatment matrix $Y_{0,\text{pre}}$; treated pre-path $y_{1,\text{pre}}$; full donor and
treated outcomes; a \emph{fixed} rank $K$; candidate grids for $\lambda_\omega$ and $\eta$; preprocessing mode.
\State \textbf{Preprocess and estimate basis (donor data only):} apply the chosen preprocessing
(Section~\ref{sec:preprocessing}) to $Y_{0,\text{pre}}$; compute $V_K$ from the processed donor matrix.
\State \textbf{Placebo-tune $(\lambda_\omega,\eta)$ at the fixed $K$:} as described above; select the pair
minimizing average placebo post-treatment squared error.
\State \textbf{Balance and predict:} solve \eqref{eq:hybrid} for $\hat\omega$ at the selected
$(\lambda_\omega,\eta)$; form $\hat Y_{1t}(0)$ (with the intercept correction of
Section~\ref{sec:preprocessing} if preprocessing is not raw); $\hat\tau_t=Y_{1t}-\hat Y_{1t}(0)$.
\State \Return $\hat\tau$, $\hat\omega$, selected $(\lambda_\omega,\eta)$.
\end{algorithmic}
\end{algorithm}

\paragraph{Sensitivity to placebo sample size.} Hyperparameters above are tuned using $4$ randomly drawn
placebo donors per replication. Table~\ref{tab:placebo} checks whether this is large enough by comparing $4$,
$10$, and all $29$ eligible donors on the baseline regime.

\begin{table}[h]
\centering
\small
\begin{tabular}{lrrrrl}
\toprule
Placebo donors & $N$ (replications) & $\bar\eta$ & $\Pr(\hat\eta=1)$ & $\Delta$(Spectral$-$SC) RMSE \\
\midrule
4    & 120 & 0.970 & 0.917 & 0.147 (0.038) \\
10   & 120 & 0.972 & 0.950 & 0.148 (0.038) \\
29 (all) & 50  & 0.994 & 0.960 & 0.194 (0.056) \\
\bottomrule
\end{tabular}
\caption{Placebo-sample-size sensitivity, baseline regime, raw preprocessing. $N$ is reduced for the
all-donor row for compute feasibility; standard errors in parentheses.}
\label{tab:placebo}
\end{table}

The preference for raw-path matching does not depend on using a small, noisy placebo sample: it is, if
anything, stronger with more placebo donors. This addresses the natural concern that four donors might be too
few to reliably estimate the placebo objective and could itself explain the endpoint selection.

\section{Simulation Study}
\label{sec:simulation}

We compare DiD, raw-path SC, truncated Spectral SC, and the hybrid estimator, each (except DiD) with
$(\lambda_\omega,\eta)$ tuned by the placebo procedure of Section~\ref{sec:algorithm} at a fixed, regime-specified
$K$, across eleven raw-input regimes chosen to favor truncation, with $400$ replications per regime.

\emph{The purpose of this comparison is not to identify the best available panel estimator.} SDID, augmented
synthetic control, generalized synthetic control, and matrix completion are established, independently
benchmarked alternatives that we do not evaluate here. The purpose is narrower: to isolate the incremental
effect of replacing raw-path SC's matching metric with a truncated, or partially truncated, spectral metric,
holding the donor-weighted-average architecture fixed.

\subsection{Design}

The baseline regime uses \eqref{eq:factor-model} with $N_0=30$, $T_0=20$, $T_1=10$, $\alpha_i,\delta_t\sim
N(0,1)$, $\ell_i\sim N(0,I_R)$, $f_t$ a random walk (increment sd $0.3$), $\varepsilon_{it}\sim N(0,0.3^2)$,
constant post-treatment effect $\tau_t=2$, $R=K=2$. Ten further regimes each change one feature: sparse
oracle donors, clustered donor loadings, treated near a hull vertex, $T_0\gg N_0$ ($N_0=10,T_0=150$),
high-frequency noise, a weak third factor tested at $K=2$ and $K=3$, post-treatment factor rotation, and a
treatment-correlated weak factor (elevated treated loading achieved via an in-hull convex combination
concentrated on the donors with the highest loading on that factor, so representability holds by
construction) tested at $K=2$ and $K=3$. Full generative detail for every regime is in
Appendix~\ref{sec:appendix}.

\begin{table}[h]
\centering
\footnotesize
\begin{tabular}{lrrrr}
\toprule
Regime & DiD & SC (tuned) & Spectral (tuned) & Hybrid (tuned) \\
\midrule
Baseline                & 0.003 (0.013)  & 0.013 (0.013)  & 0.009 (0.021)  & 0.012 (0.013)  \\
Sparse oracle donors     & 0.019 (0.032)  & 0.011 (0.015)  & 0.016 (0.025)  & 0.014 (0.016)  \\
Clustered loadings       & $-$0.060 (0.080) & $-$0.038 (0.024) & $-$0.050 (0.032) & $-$0.036 (0.024) \\
Treated at hull edge     & $-$0.085 (0.072) & $-$0.018 (0.027) & $-$0.030 (0.043) & $-$0.020 (0.027) \\
$T_0\gg N_0$              & 0.063 (0.044)  & 0.023 (0.018)  & 0.034 (0.027)  & 0.022 (0.019)  \\
High-frequency noise      & 0.007 (0.017)  & 0.027 (0.021)  & 0.028 (0.026)  & 0.026 (0.021)  \\
Weak factor, $K=2$        & 0.001 (0.014)  & $-$0.014 (0.014) & $-$0.023 (0.021) & $-$0.012 (0.014) \\
Weak factor, $K=3$        & 0.001 (0.014)  & $-$0.014 (0.014) & $-$0.022 (0.016) & $-$0.016 (0.014) \\
Post-treatment rotation   & 0.003 (0.015)  & 0.023 (0.018)  & 0.020 (0.023)  & 0.022 (0.018)  \\
Confounded weak factor, $K=2$ & $-$0.019 (0.079) & 0.009 (0.068) & $-$0.013 (0.075) & 0.015 (0.068) \\
Confounded weak factor, $K=3$ & $-$0.019 (0.079) & 0.009 (0.068) & 0.004 (0.072)  & 0.008 (0.069)  \\
\bottomrule
\end{tabular}
\caption{Bias (Monte Carlo standard error), $\tau_{\text{true}}=2$, $400$ replications per regime.}
\label{tab:bias}
\end{table}

\begin{table}[h]
\centering
\footnotesize
\begin{tabular}{lrrrr}
\toprule
Regime & DiD & SC (tuned) & Spectral (tuned) & Hybrid (tuned) \\
\midrule
Baseline                & 0.263 (0.014) & 0.268 (0.016) & 0.423 (0.019) & 0.268 (0.015) \\
Sparse oracle donors     & 0.646 (0.038) & 0.304 (0.018) & 0.509 (0.026) & 0.314 (0.021) \\
Clustered loadings       & 1.598 (0.099) & 0.479 (0.028) & 0.645 (0.036) & 0.481 (0.026) \\
Treated at hull edge     & 1.449 (0.090) & 0.543 (0.038) & 0.866 (0.048) & 0.547 (0.038) \\
$T_0\gg N_0$              & 0.884 (0.045) & 0.369 (0.034) & 0.548 (0.036) & 0.384 (0.033) \\
High-frequency noise      & 0.339 (0.013) & 0.418 (0.017) & 0.516 (0.022) & 0.427 (0.017) \\
Weak factor, $K=2$        & 0.271 (0.014) & 0.275 (0.016) & 0.425 (0.032) & 0.278 (0.016) \\
Weak factor, $K=3$        & 0.271 (0.014) & 0.275 (0.016) & 0.321 (0.021) & 0.277 (0.016) \\
Post-treatment rotation   & 0.297 (0.019) & 0.351 (0.020) & 0.455 (0.020) & 0.353 (0.019) \\
Confounded weak factor, $K=2$ & 1.588 (0.068) & 1.367 (0.071) & 1.502 (0.069) & 1.364 (0.065) \\
Confounded weak factor, $K=3$ & 1.588 (0.068) & 1.367 (0.071) & 1.440 (0.068) & 1.383 (0.068) \\
\bottomrule
\end{tabular}
\caption{RMSE (Monte Carlo standard error, $400$-fold bootstrap), $400$ replications per regime.}
\label{tab:rmse}
\end{table}

\begin{table}[h]
\centering
\footnotesize
\begin{tabular}{lrrrr}
\toprule
Regime & $\Delta$ RMSE(Spectral$-$SC) & SE & $\Delta$ RMSE(Hybrid$-$SC) & SE \\
\midrule
Baseline                & 0.155 & 0.019 & 0.001  & 0.001 \\
Sparse oracle donors     & 0.204 & 0.018 & 0.010  & 0.006 \\
Clustered loadings       & 0.166 & 0.020 & 0.002  & 0.001 \\
Treated at hull edge     & 0.323 & 0.031 & 0.004  & 0.002 \\
$T_0\gg N_0$              & 0.180 & 0.021 & 0.015  & 0.005 \\
High-frequency noise      & 0.098 & 0.020 & 0.009  & 0.006 \\
Weak factor, $K=2$        & 0.150 & 0.021 & 0.004  & 0.003 \\
Weak factor, $K=3$        & 0.046 & 0.010 & 0.003  & 0.002 \\
Post-treatment rotation   & 0.103 & 0.019 & 0.002  & 0.003 \\
Confounded weak factor, $K=2$ & 0.135 & 0.029 & $-$0.003 & 0.017 \\
Confounded weak factor, $K=3$ & 0.073 & 0.018 & 0.016  & 0.005 \\
\bottomrule
\end{tabular}
\caption{Paired RMSE differences (same simulated panels; $1{,}000$-fold paired bootstrap SE). Positive values
favor SC.}
\label{tab:paired}
\end{table}

\subsection{What the simulation shows}

\textbf{Truncated Spectral SC's RMSE disadvantage is precisely estimated and large relative to its standard
error in every regime.} Table~\ref{tab:paired} reports paired differences on the same simulated panels rather
than treating the two RMSE standard errors in Table~\ref{tab:rmse} as independent, which is the more
appropriate comparison here. The Spectral-vs-SC gap ranges from $0.046$ ($\mathrm{SE}\ 0.010$, weak factor at
$K=3$) to $0.323$ ($\mathrm{SE}\ 0.031$, hull edge) -- four to eleven standard errors from zero in every
regime, including the ones constructed to favor truncation.

\textbf{The hybrid estimator is statistically indistinguishable from tuned SC in most regimes.} Nine of eleven
paired Hybrid-vs-SC differences in Table~\ref{tab:paired} are under two standard errors from zero. The two
exceptions -- $T_0\gg N_0$ ($0.015$, SE $0.005$) and the confounded weak factor at $K=3$ ($0.016$, SE
$0.005$) -- are real but an order of magnitude smaller than the corresponding Spectral-vs-SC gaps in the same
regimes ($0.180$ and $0.073$). Consistent with this, the placebo-tuned mixing weight selects $\hat\eta=1$
(raw-path matching) in a majority of replications in every regime (Table~\ref{tab:eta}, unchanged from the
mechanism described in Section~\ref{sec:algorithm}), with $\Pr(\hat\eta=1)$ ranging from $0.67$ to $0.91$.

\begin{table}[h]
\centering
\footnotesize
\begin{tabular}{lrrrr}
\toprule
Regime & $\bar\eta$ & $\Pr(\hat\eta=0)$ & $\Pr(\hat\eta=1)$ & $\Pr(0<\hat\eta<1)$ \\
\midrule
Baseline                & 0.948 & 0.013 & 0.898 & 0.090 \\
Sparse oracle donors     & 0.948 & 0.015 & 0.895 & 0.090 \\
Clustered loadings       & 0.922 & 0.048 & 0.868 & 0.085 \\
Treated at hull edge     & 0.945 & 0.025 & 0.910 & 0.065 \\
$T_0\gg N_0$              & 0.788 & 0.123 & 0.705 & 0.173 \\
High-frequency noise      & 0.821 & 0.060 & 0.673 & 0.268 \\
Weak factor, $K=2$        & 0.928 & 0.025 & 0.873 & 0.103 \\
Weak factor, $K=3$        & 0.854 & 0.088 & 0.805 & 0.108 \\
Post-treatment rotation   & 0.940 & 0.020 & 0.898 & 0.083 \\
Confounded weak factor, $K=2$ & 0.832 & 0.098 & 0.760 & 0.143 \\
Confounded weak factor, $K=3$ & 0.788 & 0.150 & 0.730 & 0.120 \\
\bottomrule
\end{tabular}
\caption{Empirical distribution of the hybrid estimator's placebo-selected $\hat\eta$, $400$ replications per
regime.}
\label{tab:eta}
\end{table}

\textbf{The confounded weak-factor regime isolates the risk that unsupervised truncation discards a
low-variance, causally relevant direction.} With the relevant factor guaranteed in-hull by construction,
Spectral SC's RMSE falls from $1.502$ at $K=2$ (excluding it) to $1.440$ at $K=3$ (including it) -- a real,
partial recovery, and the paired Spectral-vs-SC gap shrinks correspondingly (from $0.135$ to $0.073$ SE
$0.018$--$0.029$) -- while raw-path SC and the hybrid estimator, which retain at least partial weight on all
directions regardless of $K$, are comparatively stable and remain the best performers at both values of $K$.
This is exactly the pattern Proposition~\ref{prop:bound} predicts.

We read the overall pattern as evidence for the two mechanisms formalized in
Sections~\ref{sec:underdetermination} and \ref{sec:preprocessing}, not as a general indictment of the
motivating hypothesis: under raw-input matching, low-dimensional balancing is underdetermined relative to the
donor pool (Proposition~\ref{prop:underdetermination}) and the basis itself is contaminated by unremoved
fixed effects, inflating the BLP residual in Proposition~\ref{prop:bound}. Section~\ref{sec:preprocessing}
shows removing that contamination removes most of the gap on the one regime we checked; the eleven-regime
study here does not.

\section{Discussion}
\label{sec:discussion}

\subsection{Related work}
\label{sec:related}

Matrix completion \citep{athey2021}, generalized synthetic control \citep{xu2017}, robust synthetic control
\citep{amjad2018}, and augmented synthetic control \citep{benmichael2021} use low-rank or factor-model
structure to correct or regularize SC-type estimators; SDID \citep{arkhangelsky2021} is analyzed under a
latent factor model. Synthetic Principal Component Design \citep{lu2022} uses spectral ideas for the
experimental-design problem of which units to treat, solved with a power-method-based combinatorial
optimization -- a different problem from the post-treatment score balancing studied here. A functional
extension of generalized synthetic control \citep{shao2026} uses functional principal component scores for
sparse, irregular panels within a Bayesian hierarchical model, addressing a data-sparsity problem this note
does not consider.

Harmonic Synthetic Control \citep{liu2026} is the closest prior work in spirit: it also introduces a
continuously tunable spectral allocation governed by a single parameter. It addresses a different problem by
a different mechanism, however. Harmonic Synthetic Control targets unit-specific stochastic trends in
nonstationary data by splitting variation between donor matching and a residual time-series forecasting
component, with its tuning parameter, selected by rolling-origin cross-validation, governing the allocation
between these two branches; its spectral interpretation shows how this downweights low-frequency residual
components in donor matching in favor of the forecasting branch. Our hybrid estimator introduces no
forecasting branch: $\eta$ modifies only the geometry of pre-treatment donor matching, by shrinking
discrepancies outside a donor-estimated principal subspace, with the counterfactual always a donor-weighted
average rather than a forecast. The two tuning parameters connect different pairs of objects -- theirs
interpolates between donor matching and extrapolation, ours between truncated and raw-path matching.

\paragraph{Why spectral completion is a different estimator.} A separate way to use the low-rank structure in
\eqref{eq:factor-model} is to estimate the untreated panel directly via low-rank regression on donor and
treated pre-treatment cells, then read off $\hat\tau_t=Y_{1t}-\hat\alpha_1-\hat\delta_t-\hat L_{1t}$ for
$t>T_0$. This is legitimate and correctly identified when the treatment effect is computed as a residual
outside the fitting step, but it is a version of existing matrix-completion and interactive-fixed-effects
estimators \citep{athey2021,xu2017,bai2009} rather than a variant of the donor-weighted balancing family
studied here, and we do not evaluate it in Section~\ref{sec:simulation}. Fitting the treatment effect jointly
with the low-rank component, rather than as a residual, requires an explicit restriction preventing the
low-rank term from absorbing the treatment pattern itself, since otherwise the treatment parameter does not
appear in a loss computed only on untreated cells.

\subsection{What remains open}

\begin{itemize}[leftmargin=1.6em]
  \item \textbf{The preprocessing reversal is checked on one regime.} Table~\ref{tab:preprocessing} shows
        two-way demeaning nearly eliminates the baseline regime's gap and flips the placebo-selected $\eta$;
        whether this holds across all eleven regimes, and especially in the ones with the largest raw-input
        gaps (hull edge, sparse donors), is not yet known.
  \item \textbf{$K$-selection is not addressed.} We fix $K$ by design and tune only
        $(\lambda_\omega,\eta)$; a principled joint selection of $K$ together with the other hyperparameters,
        and an understanding of why the placebo objective is not directly comparable across $K$, is left open.
  \item \textbf{Rotational identification.} $V_K$ is identified as a subspace, not as individually labeled
        components.
  \item \textbf{Inference after a data-dependent basis and hyperparameters.} $\hat V_K,\hat\lambda_\omega,
        \hat\eta$ are estimated from the same sample used to form $\hat\tau$; naive standard errors treating
        them as fixed understate uncertainty. This remains the largest methodological gap in the paper.
  \item \textbf{Staggered adoption and multiple treated units} are not addressed here.
\end{itemize}

\section{Conclusion}
\label{sec:conclusion}

We characterized two specific mechanisms governing when spectral truncation can help or hurt synthetic
control: an exact equivalence to raw-path SC at full rank, a proof that low-dimensional spectral balance is
generically underdetermined by $N_0-K-1$ degrees of freedom, and a finite-sample bound tying spectral
balancing quality to treatment-effect bias through a best-linear-predictor residual that basis-estimation
noise and omitted factors inflate. Across eleven raw-input regimes chosen to favor truncation, truncated
Spectral SC had significantly higher RMSE than tuned raw-path SC throughout, and a hybrid estimator free to
select anywhere between truncated and raw-path matching selected raw-path matching in a majority of
replications in every regime. A robustness check on the baseline regime shows this conclusion is specific to
raw-input matching: removing fixed effects before the spectral decomposition -- exactly the condition our
bound assumes -- nearly eliminates the gap and reverses the placebo procedure's preference. We present the
paper as a diagnostic result identifying when and why spectral truncation fails under raw-path matching,
alongside concrete evidence for what changes when that condition is addressed, rather than as a method ready
to replace raw-path SC.

\appendix
\section{Reproducibility Details}
\label{sec:appendix}

\paragraph{Baseline parameters.} $N_0=30$ donors, $T_0=20$ pre-treatment periods, $T_1=10$ post-treatment
periods, true factor rank $R=2$ (or $3$ in the weak-factor and confounded-weak-factor regimes), $K=2$ unless
stated otherwise. $\alpha_i,\delta_t\overset{iid}{\sim}N(0,1)$; $\ell_i\overset{iid}{\sim}N(0,I_R)$ for
non-cluster, non-sparse, non-edge, non-confounded regimes; $f_t=f_{t-1}+\xi_t$, $\xi_t\sim N(0,0.3^2I_R)$,
$f_0=0$; $\varepsilon_{it}\overset{iid}{\sim}N(0,0.3^2)$; constant post-treatment effect $\tau_t=2$.

\paragraph{Regime-specific generation.} \emph{Sparse}: treated loading $\ell_1=\sum_i w_i\ell_i$ with
$w\sim\mathrm{Dirichlet}(0.15\cdot\mathbf{1}_{N_0})$. \emph{Cluster}: two component centers
$\sim N(0,1.6^2I_R)$; each donor loading is its assigned center plus $N(0,0.3^2I_R)$ noise; treated loading is
the first center plus independent $N(0,0.3^2I_R)$ noise. \emph{Edge}: treated loading equals a uniformly
random donor's loading plus $N(0,0.05^2I_R)$ noise. \emph{High-frequency noise}: an additional term
$0.6\times(-1)^t\times N(0,1)$ is added to every donor and treated observation at every $t$. \emph{Weak
factor}: the third factor path is scaled by $0.25$ at every $t$. \emph{Post-treatment rotation}: for
$t>T_0$, the first two factor coordinates are rotated by angle $\theta\sim\mathrm{Uniform}(0.3,0.9)$ radians.
\emph{Confounded weak factor}: the third factor is scaled by $0.15$ for $t\le T_0$ and left unscaled for
$t>T_0$; donors are ranked by their third-factor loading, the top third form a pool, and the treated loading
is $\sum_iw_i\ell_i$ with $w$ supported only on that pool, $w\sim\mathrm{Dirichlet}(\mathbf{1})$ on its
support. \emph{$T_0\gg N_0$}: $N_0=10$, $T_0=150$, $T_1=20$, otherwise as baseline.

\paragraph{Optimization.} All balancing problems are solved by projected gradient descent on the simplex
(Euclidean projection via the standard $O(n\log n)$ algorithm), with step size $1/L$ for Lipschitz constant
$L=2\lambda_{\max}(X^\top MX)+2\lambda_\omega+10^{-9}$; $60$ iterations for the final applied fit, $25$
iterations during placebo tuning (verified to match a $150$-iteration solve to within $0.1\%$ of objective
value). Tie-breaking in placebo selection: grid search iterates $\lambda_\omega$ in the order listed, and (for
the hybrid estimator) $\eta$ nested within each $\lambda_\omega$, in the order listed; the first
$(\lambda_\omega,\eta)$ attaining the minimum average placebo squared error is retained.

\paragraph{Randomization and replication.} Replication $m$ of a given regime uses NumPy generator
\texttt{default\_rng(seed0 + m)} with \texttt{seed0=2000}, controlling all randomness in that replication
(panel generation and placebo donor subsampling) from a single stream. Monte Carlo standard errors for bias
are analytic ($\mathrm{sd}/\sqrt{B}$); for RMSE and for paired RMSE differences, a $400$-fold (respectively
$1{,}000$-fold) nonparametric bootstrap over replication indices, applied jointly to paired estimators to
preserve pairing.

\paragraph{Code availability.} A complete reference implementation, including the exact generative code for
all eleven regimes, the projected-gradient solver, the placebo-tuning procedure, and the preprocessing
variants, is provided as supplementary code (\texttt{simulation.py}) accompanying this submission.

\bibliographystyle{plainnat}

\end{document}